\documentclass[copyright,creativecommons]{eptcs}
\providecommand{\event}{QPL 2013}
\usepackage{breakurl,amsmath,amsthm,amssymb,xspace,xypic,enumerate}
\usepackage[usenames]{color}
\usepackage[off]{auto-pst-pdf}
\newcommand{\bit}{\begin{itemize}}
\newcommand{\eit}{\end{itemize}}
\newcommand{\ben}{\begin{enumerate}}
\newcommand{\een}{\end{enumerate}}

\theoremstyle{plain}

\newtheorem*{theorem*}{Theorem}

\theoremstyle{definition}

\newcommand{\bB}{\begin{color}{black}}
\newcommand{\e}{\end{color}}

\title{Non-locality in theories without the no-restriction hypothesis}

\author{Peter Janotta}
\author{Peter Janotta \qquad Raymond Lal
  \institute{Universit\"{a}t W\"{u}rzburg \quad University of Oxford}
  \email{peter.janotta@physik.uni-wuerzburg.de - raymond.lal@cs.ox.ac.uk}
  }
\date{}
\def\titlerunning{Non-locality in theories without the no-restriction hypothesis}
\def\authorrunning{P. Janotta \& R. Lal}

\begin{document}
\maketitle

\begin{abstract}
The framework of generalized probabilistic theories (GPT) is a widely-used approach for studying the physical foundations of quantum theory. 
The standard GPT framework assumes the no-restriction hypothesis, in which the state space of a physical theory determines the set of measurements. 
However, this assumption is not physically motivated. 
In Janotta and Lal [Phys. Rev. A \textbf{87}, 052131 (2013)], it was shown how this assumption can be relaxed, and how such an approach can be used to describe new classes of probabilistic theories.
This involves introducing a new, more general, definition of \bB maximal \e joint state spaces, which we call the \emph{generalised maximal tensor product}.
Here we show that the generalised \bB maximal \e tensor product recovers the standard \bB maximal \e tensor product when at least one of the systems in a bipartite scenario obeys the no-restriction hypothesis.
We also show that, under certain conditions, relaxing the no-restriction hypothesis for a given state space does not allow for stronger non-locality, \bB although the generalized maximal tensor product may allow new joint states.\e
\end{abstract}
% \hrule 
% \smallskip
% \noindent
% This paper summarises and extends the results in \cite{norestriction}:
% \bit
% \item Peter Janotta \& Raymond Lal (2013): \em Generalized Probabilistic Theories Without the No-Restriction Hypothesis\em. Physical Review A \textbf{87}, 052131, arXiv:1302.2632.
% \eit
% \hrule 

\paragraph{Background.}
The context of our results is the framework of \em generalized probabilistic theories (GPTs). \em
This is an operational approach for studying the physical foundations of quantum theory \cite{Barrett07}.
It is \em operational \em because a theory is defined according to the observable measurement statistics that it predicts.
Assuming only basic principles, the framework encompasses a large variety of theories, e.g.~quantum theory and classical probability theory.
Using GPTs one can examine the relationship between different physical properties, e.g.~no-cloning and nonlocality, without restricting to a particular physical theory.

\par\bigskip
For example, it has been shown that any non-classical probability theory has the following properties:  the existence of entanglement \cite{Barrett07}; for mixed states, the lack of a unique decomposition into a unique ensemble of pure states; generalizations of the no-cloning or no-broadcasting theorem \cite{Barnum08}; and, an information-disturbance trade-off \cite{Scarani06}. 
Notably, recent attempts to reconstruct quantum theory from physical axioms include the assumptions made in GPTs \cite{Masanes11, Dakic09} or very similar assumptions \cite{Chiribella11}.

\paragraph{The core elements of a GPT.}
To define a GPT, we must define at least three items: the set of states, the set of effects, and joint systems.
A \em state \em is defined as  an equivalence class of preparation procedures which all yield exactly the same measurement statistics.
Analogously we also define an \em effect \em as an equivalence class of measurement outcomes. 
Mathematically, states are represented by elements of a real vector space $V$. 
Effects are linear functionals on states, i.e. elements of the dual space $V^*$.
Applying an effect $e$ to a state $\omega$ yields the probability $p(e|\omega) = e(\omega)$ for the corresponding measurement outcome to occur when measuring the system in the state.
Both states and effects are then represented by vectors embedded in $\mathbb{R}^n$.
The application of effects on states is given by the Euclidean inner product of the respective vectors:
\begin{align}
 \label{eq:vecrep}
 e = \left( \epsilon_1, \cdots , \epsilon_n\right)^T \qquad \omega = \left( w_1 , \cdots , w_n\right)^T\\
 p(e|\omega) = e^T \!\!\!\cdot \omega = \sum_i \epsilon_i \, w_i
\end{align}
To account for uncertainty in state preparation or measurement, we represent \em mixed states \em and \em mixed effects \em by using convex combinations:
\begin{align}
 e &= \sum_i \lambda_i \, e_i \qquad \lambda_i \geq 0, \, \sum_i \lambda_i = 1\\
 \omega &= \sum_i \mu_j \, \omega_j \qquad \mu_i \geq 0, \, \sum_i \mu_i = 1
\end{align}
\em Pure states \em and \em pure effects \em are the extremal points of these convex sets, and the normalized states and effects and denoted by $\Omega^{A}$ and $E^A$ respectively.
\par\bigskip
To define composite systems, bipartite joint states are given by elements of the product space
\begin{align} 
\label{eq:productspace}
V^{AB} = V^A \otimes V^B
\end{align}
and joint effects are elements of $V^{AB*} = V^{A*} \otimes V^{B*}$.
\bB
Eq.~\eqref{eq:productspace} has important physical content, in particular, it implies \emph{local tomography}, which is the condition that bipartite states can be determined by the correlations of single-system measurements (conversely, in Ref.~\cite{Barrett07}, Barrett has shown how local tomography, in addition to other assumptions, can be used to derive Eq.~\eqref{eq:productspace} for an arbitrary GPT).
\e
However, we must define the set of joint states $\Omega^{AB}=\{\omega^{AB}\}$, and the set of joint effects $E^{AB}=\{e^{AB}\}$, such that these are \em consistent \em with the state spaces of the individual systems.
%We explain consistency as follows.
For example, note that the joint system should at least incorporate product effects and their mixtures. 
Then, for a joint state $\omega^{AB}$ to be consistent with the individual state spaces means that: (i) applying such joint effects to $\omega^{AB}$ should give probabilities (i.e.~elements of $[0,1]$);  (ii) the conditional states should be elements of the state space $\Omega^{A}$ (where by `conditional state' we mean the `post-measurement` state on one part of the joint system, conditioned on a particular measurement outcome on the other part).
This implies that the joint states form a subset of the \em maximal tensor product $\Omega^{AB}_{max}$, \em which is the largest set of such states. 
The maximal tensor product give all states respecting the positivity condition (i) with respect to product effects. 
Mathematically, positivity and therefore the maximal tensor product correspond to the construction of a so-called dual cone. 
For theories traditionally considered in the GPT framework it turns out that condition (ii) is actually equivalent to condition (i).

\paragraph{Our contribution.} 
In order to yield measurement probabilities effects are restricted to give values in the range of $[0,1]$ when applied to normalized states.
In the traditional framework of GPTs \bB (e.g.~Ref.~\cite{Barrett07})\e, the set of effects $E$ is \em not restricted any further. \em
That is, the set of effects is exactly the set of all probability-valued linear functionals on the given states.
We call this relationship between states and effects the \em no-restriction hypothesis\em, in accordance with \cite{Chiribella10}. 
It is satisfied for classical probability theory and quantum theory.
Crucially, if the no-restriction hypothesis holds, then \bB a large part of \e the theory is \em  completely determined \em by the state space, since the effect set{\bB---and hence the allowed measurements---\e}can be derived from the state space \bB(although there may still be freedom in defining the dynamics of the theory)\e.
Our contribution is to extend the framework of GPTs without the no-restriction hypothesis. 
There are two main reasons for doing so.
Firstly, the necessity of the no-restriction hypothesis is questionable from an operational perspective.
Indeed, considering the physical meaning of states and effects, there is no reason to believe that the possible preparation procedures determine possible measurements.
Secondly, this will generalize \bB the standard constructions of \e the GPT framework to cover scenarios that have not been explored previously.

\bB
\paragraph{Previous work.}
Although the framework of GPTs typically uses tools which assume the no-restriction hypothesis, there are related approaches which do not.
The formalism of test spaces developed by Foulis and Randall does not make the assumption \cite{Foulis1972,Foulis1973} (see also the overview given by Ref.~\cite{Wilce2009} and references therein). 
Barnum and Wilce use a framework of convex sets, without the no-restriction hypothesis, to describe information processing in general theories \cite{Barnum2011}, and Wilce derives the Jordan structure of quantum theory using this framework \cite{Wilce12}. Similarly, the full reconstruction of quantum theory by Chiribella, D'Ariano and Perinotti (CDP) also does not make the assumption \cite{Chiribella11}.
However, since the CDP assumptions lead to quantum theory, they determine the tensor product uniquely; in contrast, the aim of our work is to use the GPT framework to explore the range of tensor products that are consistent with any given single systems.
\e
\par\bigskip
The first step in our work is to allow $E^A$ to be defined as a set of linear functionals on $\Omega^A$, but not necessarily the \em full \em set of such functionals.
However, when defining composite systems, we are faced with a difficulty. 
When using the definition of $\Omega^{AB}_{max}$, i.e. we consider all joint states that give positive results on product effects, we obtain states $\omega^{AB}$ for which the conditional states are no longer in the single-system state spaces $\Omega^{A}$.
Referring to the requirements as described in the last section this means consistency requirement (i) does no longer imply condition (ii) for restricted systems.

Our solution is to define a \em generalized maximal tensor product \em $\overline{\Omega}^{AB}_{max} := \Omega^{\mathcal{A}B}_{max} \cap \Omega^{A\mathcal{B}}_{max}$. This is the intersection of two traditional maximal tensor products with unrestricted auxiliary systems $\mathcal{A}$ and $\mathcal{B}$ that have state spaces equal to the original systems, but effects extended to the full set of probability valued functionals.
In \cite{norestriction} we show that:
\par\smallskip
\begin{theorem*}
Let $\Omega^A$ and $E^A$ be the state space and effect space, respectively, of a GPT without the no-restriction hypothesis.
Then $\omega^{AB}\in \overline{\Omega}^{AB}_{max} $ iff
$\omega^{AB}$ has well-defined conditional states.
\end{theorem*}

The generalized maximal tensor product is a bona fide extension of the standard construction: the former reduces to the latter when the no-restriction hypothesis is assumed.
In general it lies between the traditional tensor product of the unrestricted auxiliary systems $\mathcal{A}$, $\mathcal{B}$ and the one of the original (restricted) systems $A$, $B$:
\begin{align}
 \label{eq:tensorhierachy}
 \Omega^{\mathcal{A}\mathcal{B}}_{max} \subseteq \overline{\Omega}^{AB}_{max} \subseteq \Omega^{AB}_{max}
\end{align}

\paragraph{Notable results.} 
Our construction allows us to explore new models.
\bit
\item {\it Noisy boxworld}: 
By removing the no-restriction hypothesis we can model systems with intrinsic noise, i.e. systems for which the unit measure is the \em only \em certain outcome for any state.
For example, we can introduce noise to the theory known as `boxworld' \cite{Barrett07}, which has become widely studied as it shows nonlocal correlations stronger than quantum theory \cite{Tsirelson, PRbox}.
The maximal CHSH violation $S^\lambda$ as a function of the noise parameter $\lambda$ of noisy boxworld can be shown to be $4 \, \lambda^2$.
\item {\it Self-dualized polygons}: 
The class of \em (strongly) self-dual \em systems has received much attention \cite{Barnum08,polypaper}. 
These are systems with a particular geometrical structure, shared by both classical probability theory and quantum theory.
For strongly self-dual systems states and effects can be identified with each other and thus be represented by the same mathematical objects.
We provide a self-dualization procedure that can be used to produce a self-dual theory from any theory in the GPT framework, and which requires relaxing the no-restriction hypothesis.
We show that the self-dualized versions of toy theories with state spaces given by regular polygons show quantum correlations, whereas the original models have correlations that are stronger than quantum correlations. 
\item {\it Spekkens toy theory}: 
The Spekkens theory \cite{Spekkens} is a local theory, meaning that (in the probabilistic version) it cannot violate any Bell inequalities.
However, the theory has entangled states.
This raises the question of why the Spekkens theory does not exhibit bipartite nonlocality. 
We show that a probabilistic version of the Spekkens theory does not satisfy the no-restriction hypothesis; %moreover, 
were it to satisfy this principle by taking the full dual cone, it would produce maximal nonlocal correlations.
\eit

\section*{New results}

As mentioned above the generalized maximal tensor product reduces to the traditional maximal tensor product if both systems obey the no-restriction hypothesis. 
In the following we show that the tensor products are equivalent even if only one of the systems is unrestricted.

\begin{theorem*}
 The generalized maximal tensor product $\overline{\Omega}^{AB}_{max}$ of two systems reduces to the traditional maximal tensor product $\Omega^{\mathcal{A}\mathcal{B}}_{max}$ of the auxiliary systems $\mathcal{A}$, $\mathcal{B}$, if one of the systems obeys the no-restriction hypothesis.
\end{theorem*}

\begin{proof}
Without loss of generality, assume that system $A$ is the unrestricted system.
Recall that the auxiliary system $\mathcal{A}$ extends the effect set to the full set of effects compatible with the given state space.
Since system $A$ is already unrestricted, we have $A = \mathcal{A}$ and consequently $\overline{\Omega}^{AB}_{max} = \Omega^{\mathcal{A}B}_{max} \cap \Omega^{\mathcal{A}\mathcal{B}}_{max}$.
Thus the generalized maximal tensor product is now given by the intersection of the standard maximal tensor product of auxiliary systems $\mathcal A$ and $\mathcal B$, and standard maximal tensor product of auxiliary system $\mathcal{A}$ and the original restricted system $B$.
As discussed above, the standard maximal tensor product is defined by positivity with respect to product effects, i.e. $\Omega^{\mathcal{A}B}_{max}$ are the normalized elements in $V^A \otimes V^B$ that give positive results on $E^\mathcal{A}\otimes E^B$, whereas $\Omega^{AB}_{max}$ is defined by positive results on $E^A\otimes E^B$.
In general, the restriction of a set of functionals enlarges the set of elements for which each functional yields a positive real.
Hence the restriction of the local effects from $E^\mathcal B$ to $E^B$ yields a strictly larger tensor product, 
% Hence the restriction of an effect space yields a strictly larger tensor product, 
and so $\Omega^{\mathcal{A}\mathcal{B}}_{max} \subset \Omega^{\mathcal{A}B}_{max}$.
% \cite{Short2010}.
The intersection of a set with a strictly larger set obviously gives the original set, which in this case is the standard maximal tensor product $\Omega^{\mathcal{A}\mathcal{B}}_{max}$ of the auxiliary systems.    
\end{proof}

Now, while restriction of an effect space yields fewer local measurements, there are additional joint states in the generalized maximal tensor product.
An open question is then whether restriction of local effects always limits the possible correlations to those in the unrestricted systems. 
In the following we provide a positive answer to the question for the special case where the restriction results from a linear bijection.

\begin{theorem*}
 Let $A$ and $B$ be systems with restricted effect sets $E^{A} = L^A \cdot E^\mathcal{A}$ that can be generated by a linear bijective map $L^A$ from the full unrestricted set of probability valued functionals $E^\mathcal{A}$.
 Then correlations on the generalized maximal tensor product $\overline{\Omega}^{AB}_{max}$ of the restricted systems $A$, $B$ can be at most as strong as on the maximal tensor product $\Omega^{\mathcal{A}\mathcal{B}}_{max}$ of the systems with unrestricted effect sets.
\end{theorem*}

\begin{proof}
 First let us rewrite the correlations possible in the unrestricted systems $\mathcal{A}$, $\mathcal{B}$. 
Note that for any theory there is some freedom in choosing representations of effects and states.
 In particular, one can transform the effect set by a linear bijection $L$, and counter it with a corresponding transformation $(L^{-1})^T$ on states; this will not affect measurement probabilities, since:
 \begin{align}
  \label{eq:equivrepresent}
  (L \cdot e)[(L^{-1})^T \omega] = (L \cdot e)^T \cdot (L^{-1})^T \cdot \omega = e^T \cdot (L^{-1} \cdot L)^T \cdot \omega = e^T \cdot \omega = e(\omega).
 \end{align}
 The set of bipartite correlations in a GPT result from applying each of the possible combinations of local effects in $E^\mathcal{A}$ and $E^\mathcal{B}$ to each of the joint states in $\Omega^{\mathcal{A}\mathcal{B}}_{max}$.
 Accordingly, we denote the set of all correlations as $E^\mathcal{A} \otimes E^\mathcal{B} (\Omega^{\mathcal{A}\mathcal{B}}_{max})$.
 Using Eq.~\eqref{eq:equivrepresent} this can be translated into
 \begin{align}
  \label{eq:unrestrictedcorrelations}
  E^\mathcal{A} \otimes E^\mathcal{B} \left(\Omega^{\mathcal{A}\mathcal{B}}_{max}\right) = \left(L^A \cdot E^\mathcal{A} \otimes L^B \cdot E^\mathcal{B}\right)\left[\left(L^{A,-1} \otimes L^{B,-1}\right)^T \cdot \Omega^{\mathcal{A}\mathcal{B}}_{max}\right] = E^A \otimes E^B \left(\Omega^{AB}_{max}\right). 
 \end{align}
 Consequently, the correlations of the unrestricted effects on the standard maximal tensor product of the unrestricted systems is equal to the correlations by the restricted effect set on the standard maximal tensor product of the restricted systems.

 We now compare this to the correlations that are actually possible in the restricted systems given by $E^A \otimes E^B (\overline{\Omega}^{AB}_{max})$, which uses the generalized maximal tensor product.
According to Eq.~\eqref{eq:tensorhierachy} we have $\overline{\Omega}^{AB}_{max} \subseteq \Omega^{AB}_{max}$, and together with Eq.~\ref{eq:unrestrictedcorrelations} this implies:
 \begin{align}
  E^A \otimes E^B \left(\overline{\Omega}^{AB}_{max}\right) \subseteq E^A \otimes E^B \left(\Omega^{AB}_{max}\right) = E^\mathcal{A} \otimes E^\mathcal{B} \left(\Omega^{\mathcal{A}\mathcal{B}}_{max}\right).
 \end{align}
\end{proof}

Note that this result only applies to restrictions by a linear bijection.
The general question if non-linear restrictions of effects can produce correlations exceeding the unrestricted case is still open.

\paragraph{Future work.}
There are several possible further avenues, but it would be interesting in particular to understand the relationship between our constructions and the formalism of categorical quantum mechanics \cite{AC}, in which the definition of composite systems is the primary formal device.

\paragraph{Acknowledgements.}
PJ is supported by the German Research Foundation (DFG).
RL is supported by the Templeton Foundation.

\small
\bibliographystyle{eptcs}
\bibliography{QPL2013}
\end{document}